\documentclass[preprint,12pt,authoryear]{elsarticle}

\makeatletter
\def\ps@pprintTitle{%
  \let\@oddhead\@empty
  \let\@evenhead\@empty
  \let\@oddfoot\@empty
  \let\@evenfoot\@oddfoot
}
\makeatother

\usepackage[a4paper]{geometry}
\usepackage{hyperref}
\usepackage{url}
\usepackage{amsmath}\usepackage{amsfonts}
\usepackage{amssymb}
\usepackage{graphicx}
\usepackage{algpseudocode}
\usepackage{algorithm}
\usepackage{caption}
\usepackage{subcaption}
\usepackage{tikz}
\usetikzlibrary{arrows,automata}

\newtheorem{theorem}{Theorem}

\newtheorem{corollary}[theorem]{Corollary}

\newtheorem{definition}[theorem]{Definition}

\newtheorem{proposition}[theorem]{Proposition}

\newenvironment{proof}[1][Proof]{\noindent\textbf{#1.} }{\ \rule{0.5em}{0.5em}}

\newcommand{\blind}[1]{#1}

\begin{document}

\begin{frontmatter}
\title{Close Communities in Social Networks: Boroughs and \textit{2}-Clubs\tnoteref{t1}}
\tnotetext[t1]{This research was  supported by \blind{the Netherlands Organization for Scientific Research (NWO)} under project number \blind{380-52-005 (PoliticalMashup)}.}

\author[uva]{\blind{S. Laan}}\ead{\blind{MSc.Steven.Laan@gmail.com}}
\author[uva]{\blind{M. Marx}}\ead{\blind{maartenmarx@uva.nl}}
\author[uva]{\blind{R.J. Mokken\corref{cor1}}}\ead{\blind{mokken@uva.nl}}

\cortext[cor1]{Corresponding author}
\address[uva]{\blind{ISLA, Informatics Institute, University of Amsterdam, Science Park 904, 1098XH, Amsterdam, The Netherlands}}

\begin{abstract}
The structure of close communication, contacts and
association in social networks is studied in the form of maximal subgraphs of diameter 2 (\emph{2}-clubs), corresponding to three types of close communities: hamlets, social circles and coteries. The concept of borough of a graph is defined and introduced. 
Each borough is a chained union of \emph{2}-clubs of the network and any \emph{2}-club of the network belongs to one borough. Thus the set of boroughs of a network, together with the \emph{2}-clubs held by them, are shown to contain the
structure of close communication in a network. Applications are given with examples from real world network data.\bigskip
\end{abstract}


\end{frontmatter}

 \section{\textbf{Introduction}}

The last decade has produced an increasing volume of  methods and
algorithms to analyze community structure in social and other networks, as
witnessed by an abundance of recent reviews 
\emph{e.g.}
\citep{Girvan2002,Newman2004,Balasundaram2005,Palla2005,Reichhardt2006,Blondel2008,Leskovec2008,Porter2009,Fortunato2010,Xie2013}.

In this paper we study the structure of close communication, contacts and
association in networks, as represented by simple graphs. \emph{Close
communication} is defined here as contact between nodes at distances of at
most 2, that is by direct contact or by at least one common neighboring
node. \ Such communication is associated with closely-knit groups like
cliques, coteries, peer groups, primary groups and face-to-face communities,
such as small villages and artist colonies. Considered as dense social
networks they can form powerful sources of social capital and support for their
members and serve both quick internal diffusion of social innovation as
well as speedy epidemiological contamination from outside sources.

The parts of a network where close communication can take place are marked by
overlapping subsets of nodes, which all are neighbors of each other or have a
common neighbor in the same subset. These correspond to graphs with a diameter of at most two.

In the following sections we shall characterize this structure and indicate
ways to detect these in social networks.

 Mokken \citep{Mokken1979,Mokken2008} introduced the concept of \emph{k-clubs} of a graph as \emph{maximal} induced subgraphs of diameter at most $k$ of a simple connected graph $G$: 'maximal' in the sense that there is no larger induced subgraph of diameter $k$ which includes them. He also showed that close community networks, in the form of simple graphs of diameter at most two (\textit{2-clubs}), come in three distinct types: \emph{coteries, social circles} and \emph{hamlets}, respectively.\citep{Mokken1980}\
 
Accordingly, the \emph{2-clubs} of a simple graph or network $G$ cover the areas of close communication in that network consisting of non-inclusive, possibly mutually overlapping
\emph{coteries}, \emph{social circles} and \emph{hamlets}.

In the following sections this system of close communication is studied further and we show that it consists of a set of disjoint containers of nonseparable \emph{2-}clubs,  \textit{i.e.} subgraphs that we call \emph{boroughs}, each of
which is formed by a set of edge-chained \emph{2}-clubs (hamlets, social circles
and coteries) of the network $G$. Each (nonseparable) \emph{2}-club of $G$ is included in
exactly one borough of $G$ and each borough consists of a nonseparable union
of overlapping \emph{2}-clubs of $G$. Consequently this system of close
communication of a network can be analyzed by studying its boroughs and the \emph{2}-clubs within each or selected boroughs.
The final sections show applications  with some real networks and conclude with a discussion.

\section{\textbf{Concepts and notation}}

As the representation and analysis of networks will be in terms of simple
graphs, we will summarize the necessary concepts and notation here. (For
standard graph-theoretic background see \emph{e.g.}\citep{Harary1969,Harary1994,Wasserman1994,Diestel2005}.

A social network will be represented by a \emph{simple}
graph, \emph{i.e.} an undirected graph $G=G\left(  V,L\right)  $, without
loops or multiple edges, where $V=V\left(  G\right)  $ is its set of nodes and
$L=L\left(  G\right)  $ is its set of edges $\left(  u,v\right)  ;u,v\in
V\left(  G\right)  ,$ joining nodes $u$ and $v$ in $G.$ Two nodes $u$ and $v$
are \emph{adjacent} if the edge $\left(  u,v\right)  \in L\left(  G\right)  ;$
notation $uv$. An edge $\left(  u,v\right)  $ is \emph{incident} with its
endnodes $u$ and $v$. Let $\left\vert V\right\vert $ denote the size of $G$,
\emph{i.e.} the number of its nodes and $\left\vert L\right\vert $ its number
of edges. Unless specified otherwise we shall assume $\left\vert V\right\vert
=n$ and $\left\vert L\right\vert =m$.\medskip

A subgraph $H=G(V',L')$ of a graph $G$ is a graph such that all its nodes and its edges
are in $G$:\
\[
V'  \subseteq V\left(  G\right)  \text{ and }L' \subseteq L\left(G\right)  \text{.}
\]
If $H$~is a subgraph of $G$ then $G$ is called the \emph{supergraph} of
$H$.\bigskip

If a subgraph \textit{G(V')} of $G$, with $ V'\subseteq V $, has \textit{all} edges $ (u,v) $ with $ (u,v) $ $\epsilon $ $ L $, then \textit{G(V')} is an \textit{induced subgraph} of \textit{G}. Unless stated otherwise, we shall use the term subgraph to denote an induced subgraph and consider only subgraphs \textit{G(V')} with at least three nodes and three edges.\\

A \textit{path} $ P_{l}$ is a sequence of distinct adjacent nodes of G, $  \left\lbrace u,x_{1},x_{2},..,x_{l-1},v\right\rbrace $, and consecutive incident edges $  \left\lbrace (u,x_{1}), (x_{1},x_{2}),..,(x_{l-1},v)\right\rbrace $ joining two nodes $ u $ and $ v $ in $ G $.\ 
 Its \emph{length} is the number $ l $ of its
edges. A \emph{chordless path} $P_{l}$ is a path such that no two
non-successive nodes ($\left\vert i-j\right\vert \neq1$) are adjacent. Two
nodes are connected in $G$ if there is a path joining them. The \emph{distance
}$d_{G}\left(  u,v\right)  =d\left(  u,v\right)  $ \ between two \ nodes $u$
and $v$ of $G$ is the length of a shortest path joining $u$ and $v$ in $ G $. If the
nodes are not connected then $d\left(  u,v\right)  $ is defined $\ $as
$\infty$.
$  $
The diameter $dm\left(  G\right)  $ of $G$ is the largest distance
between nodes in $G$.\\

A \textit{k-club}  in $ G $ is an induced subgraph of \textit{G} of diameter at most $ k $ \citep{Mokken1979,Mokken2008}. It is a \textit{maximal}  \textit{k-club} of $ G $ if there is no larger \textit{k-club} in $ G $ which includes it. A \textit{maximum} \textit{k-club} is one with the largest size in \textit{G}.\\ Unless stated otherwise in this paper, \textit{k-club}, respectively \textit{2-club}, \textit{of G}  will denote a \textit{maximal} \textit{k}-club, or \textit{2}-club, because \textit{k}-clubs in \textit{G}, which are included in larger \textit{k}-clubs, are not of primary interest here. We shall refer to graphs of diameter at most $ 2 $ as \textit{2-clubs}.

A cycle of $G$ is a closed path in $G$ where each node is both a starting and
an endnode in that path and no node occurs more than once. Its length $\left(l\right)  $ is the number of edges (or nodes) of it. The smallest cycle $\left(C_{3}\right)  $, a triangle, has length three. \emph{\ }A graph with cycles is \emph{cyclic. }A cycle which is an induced subgraph of $ G $ is called a chordless cycle or (for $ l > 3 $) a hole of \textit{G} \citep{Nikolopoulos2007}.\\ Unless stated otherwise 'cycle' will denote a $ C_{3} $ or a hole of $ G $.

Any edge $\left(  u,v\right)  $ of a cyclic graph can be a part of multiple cycles, to be denoted as its cycles. Its removal from $ G $ can increase some distances between nodes in $ 
G $. For instance, the distance $ d(u,v) $ then increases from $ 1 $ to $ l-1 $ if the length of its shortest cycle is a $ C_{l} $.\footnote{For related points see \citep{Granovetter1973,Everett1982}}

The \emph{degree} $ d_{G}(u) =d\left(  u\right)  $ of a node $u$
is the number of edges incident with $u$, which in a simple graph is equal to
its number of neighbors. An isolated node has degree $0$. A \emph{pendant} is a
node with a single neighbor and has degree $1.$ The \emph{average degree} of a
graph is $\bar{d}_{G}=\frac{2m}{n}$.\

For a connected graph $G$ the degree $ d_{G}\left(  u\right) $ of a node $ u $ takes values in the interval \[1\leq\delta\leq d_{G}(u)\leq\Delta\leq |V\left(  G\right)|  -1= n-1\text{.}\]
where $\delta$ and $\Delta$ are the minimum and maximum degrees of $G$.

A component of a graph is a maximal connected subgraph. A cutpoint is a node, the removal of which increases the number of components, and a bridge is an
edge with the same property. A graph with cutpoints is called \emph{separable}. Connected graphs without cutpoints are called \emph{nonseparable (n-s)} or, alternatively, \emph{2}-connected or bi-connected, and have minimum degree $\delta$ $\geq2.$ Hence it has no pendants. A \textit{bicomponent} of a graph is a maximal biconnected subgraph and is part of a component of that graph.
Such a (sub)graph is also called a \emph{block}.\ 
Unless specified otherwise we shall assume the simple graph and network to be
connected, thus consisting of a single component.

A connected graph with no cycles (\emph{acyclic}) is called a \emph{tree}.
Each connected graph has a spanning tree, \emph{i.e.} an acyclic subgraph on all nodes of the graph.

A \emph{spanning tree} of a graph $G$ has all nodes of $G$. Every
connected graph has at least one spanning tree. A \emph{shortest spanning
tree} (s.s.t.) of $G$ is a spanning tree with the smallest diameter.

In a \textit{complete} graph $K_{l}$  all $l$ nodes are mutually adjacent and its diameter $dm\left(  K_{l}\right)  =1$. A \textit{clique} of a graph $ G $ is an induced subgraph of $ G $ which is a complete graph. It is a \textit{maximal} clique of $ G $ if there is no larger clique in $ G $ containing it. A \textit{maximum} clique of $ G $ is  one with the largest size.
\\

For a node $u$ $\in V\left(  G\right)  $ we distinguish:
\begin{itemize}
\item the $k$-\emph{neighbors} of $u$:\emph{\ }$V_{k}(u)$ is the set of all nodes
$v\in V(G)$ with\\ $d(u,v)=k, k=1,...,$ $dm(G)$. Note that $u\notin V_{k}(u)$.

\item the $k$-\emph{neighborhood} of\emph{\ }$u$:\emph{\ }$N_{k}(u)=\bigcup
\limits_{i=1}^{k}$ $V_{i}(u)$, the set of all nodes $v\in V(G)$ with
$d(u,v)=1,2,...,k\leq$ $dm(G)$. Note that $u\notin N_{k}(u)$. The \emph{closed
}$k$-\emph{neighborhood} of\emph{\ }$u$ is defined as $\bar{N}_{k}(u)=N_{k}(u)\cup\left\{  u\right\}  $.
\end{itemize}

\medskip

The $k$-\emph{degree} $d_{G}^{\left(  k\right)  }\left(  u\right) = d_{k}\left(  u\right)  =\left\vert N_{k}(u)\right\vert $ is the size of the
$k$-neighborhood of $u$.
We shall in particular consider the $2$-\emph{neighborhoods}
of nodes \textit{u} of $G$: $N_{2}(u)$ and $\bar{N}_{2}(u)$. The $2$-degree of nodes \textit{u} of $G$, $( d_{2}\left(  u\right): u \,\in \,V)$ are bounded by minimum and maximum $2$-degrees
given by $\delta_{2}\left(  G\right)  =$ $\delta_{2}$ and $\Delta_{2}\left(
G\right)  =$ $\Delta_{2}$ in the interval \[2\leq\delta_{2}\leq\Delta_{2}\leq
|V\left(  G\right)|  -1=n-1.\]

The $k$-\emph{ego-network} of $u$ in $G$ is the subgraph $G(\bar{N}_{k}(u))$ induced by the nodes of its closed $k$-neighborhood, to be denoted as  $\mathcal{E}_{k}^{G}(u)$ (or $ \mathcal{E}_{k}(u) $ if the context is clear).                                                 
Note that $u$ thus is
part of its ego-network and is its central ego. We shall in particular consider
the ego-networks $\mathcal{E}(u)=\mathcal{E}_{1}(u)$ and $\mathcal{E}_{2}(u)$ of nodes, with sizes
$\left\vert \bar{N}_{1}(u)\right\vert $ and $\left\vert \bar{N}_{2}
(u)\right\vert $.

\emph{Twinned} ego-networks occur when the ego-networks of two or more nodes
$u_{0},u_{1},...$ coincide:
\[
\mathcal{E}(u_{0})=\mathcal{E}(u_{1})=...
\]
thus forming a single ego-network with multiple egos $u_{0},u_{1},..$. Its ego
nodes are called \emph{twinned nodes} or just \emph{twins}.
The set of
central egos $\left\{  u_{0},u_{1},...\right\}  $ of a twinned  ego-network forms a \textit{clique} and is called
its \emph{center}.  Each center $\left\{  u_{0}
,u_{1},...\right\}  $ can be represented by one of its ego-nodes $u_{0}$. We can
accordingly define a reduced node set $V^{\left(  c\right)  }\left(  G\right)
$ as the set of ego-nodes of $G$, including just a single ego node $u_{0}$ from
each twinned ego-network.

Observe that if $\mathcal{E}(u_{0})=\mathcal{E}(u_{1})=...$ then $\mathcal{E}_{k}(u_{0})=\mathcal{E}_{k}(u_{1})=...$
for $k \geq 2$, so that for twinned ego nodes all their $k$-ego-networks are twinned ego-networks.

 Moreover, as nodes can belong to various (sub)graphs, it should be stressed that the relevant ego-network $ \mathcal{E}_{k}^{H}(u) $ of a node $ u, u\in $  $V(H)\subseteq\, V(G), $ is determined by the particular (sub)graphs $ H $ of $ G $  for which they are induced.

In this paper close communities, such as acquaintance networks, are studied in
the form of simple (sub)graphs of diameter at most 2.\footnote {\label{not:k-Bor} The theorems and corollaries in this paper can be extended and proven for the general case of diameter \textit{k} \textit{e.g.} \textit{k}-clubs, k-clubs, \textit{k}-boroughs, etc.  Given the focus
of this paper on diameter 2, and to simplify presentation and analysis accordingly, we shall formulate our results mainly for this special case.} In a close community of that type the closed 2-neighborhood of each of its members covers its complete population: the 2-ego-network ($ \mathcal{E}_{2}(u) $) of each node $ u $ coincides with the network of that community.

\section{Close communities as \textit{2}-clubs of a network}

Close communities are closely-knit in the sense that every pair of its members are neighbors or has at
least one common neighbor, where the neighboring relationship represents a
durable or stable acquaintance, contact or association relation. They are modeled by \emph{2-clubs}: graphs of diameter at most 2. Mokken \citep{Mokken1980}
characterized such graphs in terms of the diameter $ 2 $, $ 3 $, or $ 4 $ of a shortest spanning tree
(s.s.t.) \emph{i.e.} a spanning tree with smallest possible diameter (assuming $ |V(G)|\geq 3 $), as a measure of
their compactness.

\subsection{Close communities: hamlets, social circles, coteries}

\textit{2}-Clubs can only have s.s.t.'s with diameter 2, 3, or 4
corresponding to the following three types:\bigskip

  1: \emph{Coteries}. A \emph{coterie} is a \textit{2}-club with a 
shortest spanning tree of diameter 2, corresponding to a spanning
star, formed by one central node $u_{0}$, which is adjacent to all other
nodes. Hence a coterie is the ego-network $\mathcal{E}\left(  u_{0}\right)  $ of its
central ego $u_{0}$. When a coterie has several s.s.t's, each with central
nodes $u_{0},u_{1},...,$ it is a twinned ego-network with twinned ego nodes
$u_{0},u_{1},...$, with the extreme case of a clique (diameter 1) were
each node is the center of a spanning star. Thus a clique is a special case of
a coterie. The smallest separable coterie is a tree of three points.
The smallest nonseparable coterie is $C_{3}$, a triangle (diameter 1).\\

2: \emph{Social circles.} A \emph{social circle} is a \textit{2}-club with
an s.s.t. of diameter 3. Because every spanning tree with odd diameter has a
center consisting of two adjacent nodes \citep{Harary1969,Harary1994}, a social circle
has at least one \emph{central pair} of neighbours (adjacent nodes) $u_{0}v_{0}$, which
together are adjacent to all the other nodes (a \emph{coupled star}; See Fig. 4 in \citep{Mokken1980}). Hence a social circle is a \textit{2}-club, such that there is at least one (central) edge $\left(  u,v\right)  $ with $V_1(u)\cup V_1(v)=V(G)$. 
The smallest social
circle is $ C_{4}$, a rectangle (diameter 2).\medskip

3: \emph{Hamlets}. A \emph{hamlet} is a \textit{2}-club with an s.s.t. of
diameter 4. Such an s.s.t. (a \emph{double, 2-step, star}; Fig. 5 in \citep{Mokken1980}),
can be obtained in two steps from any node of the graph as its center.
Hence a hamlet has no central node or a spanning star, nor a central adjacent
pair of nodes on a coupled star. Each node can be used as the starting node and center of an
s.s.t. The smallest
hamlet is $C_{5}$, a pentagon (diameter 2).\\

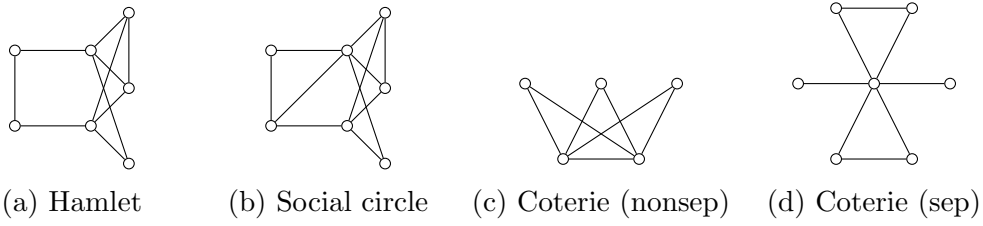
\begin{figure}[h]
\centering
\begin{subfigure}[b]{0.22\textwidth}
\centering
\begin{tikzpicture}[scale=0.5]
\tikzstyle{every state}=[draw,circle,fill=white,minimum size=4pt,
                            inner sep=0pt]
\node[state] (a) at (0,3) {};
\node[state] (b) at (0,1) {};
\node[state] (c) at (2,1) {};
\node[state] (d) at (2,3) {};
\node[state] (e) at (3,4) {};
\node[state] (f) at (3,2) {};
\node[state] (g) at (3,0) {};

\path (a) edge node {} (b)
          edge node {} (d)

      (c) edge node {} (b)
          edge node {} (g)
          edge node {} (f)
          edge node {} (e)

      (d) edge node {} (f)
          edge node {} (g)

      (e) edge node {} (f)
          edge node {} (d);
\end{tikzpicture}
\caption{Hamlet}
\end{subfigure}
\begin{subfigure}[b]{0.22\textwidth}
\centering
\begin{tikzpicture}[scale=0.5]
\tikzstyle{every state}=[draw,circle,fill=white,minimum size=4pt,
                            inner sep=0pt]
\node[state] (a) at (0,3) {};
\node[state] (b) at (0,1) {};
\node[state] (c) at (2,1) {};
\node[state] (d) at (2,3) {};
\node[state] (e) at (3,4) {};
\node[state] (f) at (3,2) {};
\node[state] (g) at (3,0) {};

\path (a) edge node {} (b)
          edge node {} (d)

      (c) edge node {} (b)
          edge node {} (g)
          edge node {} (f)
          edge node {} (e)

      (d) edge node {} (f)
          edge node {} (g)
          edge node {} (b)

      (e) edge node {} (f)
          edge node {} (d);
\end{tikzpicture}
\caption{Social circle}
\end{subfigure}
\begin{subfigure}[b]{0.25\textwidth}	
\centering
\begin{tikzpicture}[scale=0.5]
\tikzstyle{every state}=[draw,circle,fill=white,minimum size=4pt,
                            inner sep=0pt]
\node[state] (a) at (1,0) {};
\node[state] (b) at (3,0) {};
\node[state] (c) at (0,2) {};
\node[state] (d) at (2,2) {};
\node[state] (e) at (4,2) {};

\path (a) edge node {} (b)
          edge node {} (c)
          edge node {} (d)
          edge node {} (e)

      (b) edge node {} (c)
          edge node {} (d)
          edge node {} (e);
\end{tikzpicture}
\caption{Coterie (nonsep)}
\end{subfigure}
\begin{subfigure}[b]{0.22\textwidth}	

\centering
\begin{tikzpicture}[scale=0.5]
\tikzstyle{every state}=[draw,circle,fill=white,minimum size=4pt,
                            inner sep=0pt]
\node[state] (a) at (2,2) {};
\node[state] (b) at (1,0) {};
\node[state] (c) at (3,0) {};
\node[state] (d) at (4,2) {};
\node[state] (e) at (3,4) {};
\node[state] (f) at (1,4) {};
\node[state] (g) at (0,2) {};

\path (a) edge node {} (b)
          edge node {} (c)
          edge node {} (d)
          edge node {} (e)
          edge node {} (f)
          edge node {} (g)

      (b) edge node {} (c)

      (f) edge node {} (e);
\end{tikzpicture}
\caption{Coterie (sep)}
\end{subfigure}
\caption{The three types of 2-clubs}
\label{fig:2Ctypes}
\end{figure}

We summarize the above observations in the following theorem.
\begin{theorem}\label{thm:types2c} Three types of \textit{2}-clubs.
\item (i) A \textit{2}-club is either a coterie, a social circle or a hamlet.
\item (ii) Social circles and hamlets are nonseparable.
\item (iii) Coteries can be separable or nonseparable.
\end{theorem}

Examples of these types are given in Figure \ref{fig:2Ctypes}. If a \textit{2}-club is separable, then it must have a single spanning star and a corresponding single central node, which is its cutpoint \citep{Mokken1980}, p.6). Hence it is a separable ego-network and coterie. Thus only \textit{coteries} can be separable, and then have a single central node, which is its cutpoint (see Figure \ref{fig:2Ctypes},(d) ), while twinned coteries ares always nonseparable (cf. Figure \ref{fig:2Ctypes}(c)) \\

The smallest nonseparable examples of each type are the cycles $C_{3}$ (coterie), $C_{4}$ (social circle), $C_{5}$ (hamlet).\footnote{ Note that $ C_{3} $ has diameter 1 and is a 1-club also.}.
\\

The next theorem shows how the types of nonseparable \textit{2}-clubs are formed by these cycles.
\medskip

\begin{theorem}
\label{thm:k-edges}
Let $ G $ be a nonseparable \textit{2}-club, then for each edge of $ G $ its shortest cycle is  $ C_{3}$, $ C_{4} $, or $ C_{5} $, and
\item (\textit{i}): if $ G $ is a coterie: for each edge this is a triangle $ C_{3} $;
\item (\textit{ii}): if $ G $ is a social circle: for each edge this is a triangle $ C_{3} $ or a rectangle $ C_{4} $;
\item (\textit{iii}): if $ G $ is a hamlet: for each edge this is a triangle $ C_{3} $, a rectangle $ C_{4} $, or a pentagon $ C_{5} $.

\end{theorem}
	
\begin{proof}
No edge of $ G $ can be on a shortest cycle $ C_{k} $ for $ k > 5 $ because then its diameter would be larger than 2.\\
(\textit{i}) \textit{coterie}: $ G $ has a shortest spanning tree (s.s.t.) consisting of a central node adjacent to all other nodes of $ G $. Hence all other edges joining nodes of $ G $ are on at least one triangle $ C_{3} $ of $ G $;\\
(\textit{ii}) \textit{social circle}: $ G $ has an s.s.t. consisting of a central edge $\left(  u,v\right)  $ with $V_1(u)\cup V_1(v)=V(G)$. Hence any other edge joining nodes of $ G $ forms either a triangle with node \textit{u}, \textit{v} or edge $ (u,v), $  or a square  on the edge $ \left( u, v\right)  $;\\
(\textit{iii}) \textit{hamlet}: $ G $ has an s.s.t. such that all other edges joining nodes of $ G $ are on triangles, squares, or pentagons.
\end{proof}\\

We shall call these cycles \textit{basic cycles}, and shall denote the set of cycles $ \lbrace C_{3}, C_{4}, C_{5} \rbrace $ of a \textit{2}-club as its set of basic cycles $ \mathcal{C}_{\left[ 3,5\right] }$, or just its \textit{basic set}.\\ More general: we shall call the set of cycles of type $ \lbrace C_{3}, C_{4}, C_{5} \rbrace $ of a graph $ G $ its \textit{set of basic cycles}, or \textit{basic set} $ \mathcal{C}_{\left[ 3,5\right] }$, which form its nonseparable \textit{2}-clubs, as we shall see below. Moreover, we shall denote by \textit{basic edge of G} any edge $ (u,v) $ of \textit{G} which is on at least one basic cycle of \textit{G}.

\subsection{The \emph{2}-clubs of a graph or network}

 A \emph{k-club} of a simple graph $G$ is a maximal induced subgraph of diameter at most \emph{k} \citep{Mokken1979,Mokken2008}. It was introduced as a generalized clique concept to distinguish it
from $k$-cliques of a graph, which were defined as maximal clusters of nodes of a graph
within \emph{distance }$k~$ in the distance matrix of that graph \citep{Luce1950}.
However, considered as subgraphs $k$-cliques need not be connected, whereas
$k$-clubs, due to the diameter restriction, are warranted to be connected
subgraphs. The $k$-clubs of a network correspond to locally autonomous
subnetworks in the sense that their nodes can communicate or reach
each other within distance \textit{k} along paths involving only nodes of the
$k$-club, and not outside nodes in the larger supernetwork, as would be the
case for $k$-cliques. Occasionally a $k$-club happens to be a $k$-clique as
well, in which special case it is called a \emph{k-clan}.\footnote{ This case was called a sociometric clique by Alba \citep{Alba1973}.}

Recently \emph{k-}clubs have found interest and applications in many network oriented disciplines, such as telecommunication \citep{Balasundaram2006}, biology \citep{Balasundaram2005}, genetics \citep{Pasupuleti2008} , forensic data mining \citep{Memon2006}, web search \citep{Terveen1999}, graph mining \citep{Cavique2009}, and language processing \citep{Miao2004,Gutierrez2011}.\\
\

Above we defined close communication, contact or association in connected networks and graphs as connectedness along paths of at most length 2 and, accordingly, \emph{close communities} as \emph{\textit{2}-clubs} of a graph or network.\\
As such the concept of \textit{2}-clubs of a network is of central importance for the analysis of close communities and close communication structures
in networks. In that analysis, however, the first type of \textit{2}-club (ego-network or coterie) is of subordinate interest compared to the other two, the social circles and hamlets.\\
\paragraph{\textit{Three types, three levels of close communication}} These three types of \textit{2}-club imply different perspectives or levels of local communication:\
\begin{itemize}
\item \textit{Level 1} and most local (ego-network or \textit{coterie}). Coteries in a graph are rather restricted forms of close communities, as they correspond to just all the ego-networks of a graph, which define and span that graph. There is a central ego node and all close communication is possible via that ego within its ego-network: tightly meshed, involving triangles only.
Thus every ego-network $\mathcal{E}_{1}(u) $ of $ G $ is a rather trivial coterie in \textit{G}
with its ego node(s) $u$ as the center of a spanning star joining all its neighbors.
However, only when it is not included in a larger \textit{2-club} of $G$, and therefore  maximal, it is a coterie \textit{of} $G$.

\item \textit{Level 2}: intermediate (\textit{social circle}). There is no central node but instead at least one central pair of nodes: two adjacent neighbors, forming a central edge, which together are adjacent to the other nodes of the social circle. All close communication is possible via two central neighbours within (parts of) their ego-networks: more loosely meshed, along triangles and rectangles.

\item\textit{Level 3} and widest (\textit{hamlet}). In hamlets there are no central nodes or central pairs of nodes and all close communication is between (parts of) ego-networks, widely meshed along triangles, rectangles, pentagons.
\end{itemize}
 
Thus \emph{coteries} are limited forms of close communities, as they correspond to (maximal) ego-networks of $ G $, varying from stars to cliques. \\
Moreover, any ego-network which, as an induced subgraph of \textit{G}, is separable, \textit{i.e.} has a cutpoint, is a coterie of \textit{G}, because any subgraph of $ G $ containing that ego-network will have diameter larger than 2, as can be verified easily.\

In particular. every pendant of  $ G $ promotes the ego-network of its single neighbor to a coterie
of $ G $ (cf Figure \ref{fig:2Ctypes}(d)). Again, long isolated paths $ P_{l} $ in $ G $ are formed by overlapping path segments $ P_{3} $, which are overlapping separable coteries of $ G $, consisting of one central node and two pendant neighbors.\

As a consequence any graph or network will also show a multitude of rather trivial (separable) coteries.
Hence, from a perspective of close communication, the coteries of a network $G$ are relatively elementary, if not trivial, \textit{2}-clubs as such, when compared with the social circles and hamlets of $G$. 
They are confined to the level of local communication \textit{within} their ego-network, whereas the hamlets and social circles involve the wider levels of close communication \textit{between} and across (parts of) different ego-networks of $G$.\\
Our main focus will be on the more proper types of \textit{2}-clubs: social circles and hamlets. Moreover, we will  consider only \textit{2}-clubs with at least three nodes and three edges.\\
 
\subsection{The boroughs of a graph or network}

The nonseparable \emph{2}-clubs of a network or graph $G$ are  contained  in the bicomponents of $G$. We shall now  introduce a new type of maximal subgraph of $G$, always contained in a single bicomponent of $G$, which we call a \emph{borough} of $G$.\\ 
The main result of this section is that each borough contains  nonseparable \emph{2}-clubs of $G$, that each nonseparable \emph{2}-club of $G$ is a nonseparable \textit{2}-club of exactly one borough of $G$, and that both nonseparable \textit{2}-clubs and boroughs consist of \textit{edge chained} basic cycles: $ C_{3} $ (triangles), $ C_{4} $  (squares) and $ C_{5} $ (pentagons).\\

  Two cycles of a graph $ G $ are said to be \emph{edge connected} when they share at least one common edge.
  
 More generally: a pair of basic cycles $C^{\alpha}, C^{\beta} \in\mathcal{C}_{\left[3,  5\right]  }$ of \textit{G} is \emph{edge chained} in \textit{G} if they are edge connected, or there is a sequence of edge connected basic cycles $C^{1},...C^{i},...C^{c}\in\mathcal{C}_{\left[3,  5\right]  }$ of \textit{G}, such that $C^{\alpha}$ is edge connected with $C^{1}$, and $C^{c}$ with $ C^{\beta} $, and each intermediate consecutive pair of basic cycles $C^{i}$,
  $C^{i+1}\in\mathcal{C}^{\left[3,  5\right]  }$ is edge connected as well.\\
  
  Using these concepts we can now define a \textit{borough} in a graph.\footnote{ \citep{Batagelj2007} introduced \textit{k}-gonal connectedness of cycles. Our edge chained connection of basic cycles corresponds to their 5-gonal connectedness.}
    
\begin{definition}\label{Def:Borgh}
A borough in a graph \textit{G} is an induced  subgraph of \textit{G} such that each of its edges is on a shortest cycle $ C_{s}\,\in\,\mathcal{C}_{[3,5]}$, the basic set of \textit{G}, and all pairs of its basic cycles are edge chained in \textit{G}.
\end{definition}
A borough of \textit{G} is \textit{maximal}, denoted as a \textit{borough} (\textit{B}) \textit{of} \textit{G}, if it is not contained in a larger borough of \textit{G}. Unless specified otherwise we shall consider only maximal boroughs of $ G $.\\

A nonseparable \textit{2}-club is a special case of a borough as stated in the next proposition: 
 
\begin{proposition}\label{thm:2c}
A nonseparable \textit{2}-club is a borough of diameter at most 2.
\end{proposition}
\begin{proof} Let \textit{G} be a \textit{2}-club. By Theorem \ref{thm:k-edges}  every edge of \textit{G} lies on a basic cycle. Let $ C^{i} $ and $ C^{j} $ be two basic cycles of \textit{G} which are not edge connected in \textit{G} and let $e_{i} $ and $ e_{j} $ be two edges of $ C^{i} $ and $ C^{j} $ respectively. As \textit{G} has diameter 2 and is nonseparable, the endnodes of both edges must have common neigbours  in \textit{G}, on  intermediate edge connected basic cycles  which are edge connected with $ C^{i} $ and $ C^{j} $, thus establishing the edge chained connection between them. Hence $ G $ is a borough.\end{proof}\\

Analogously, two nonseparable \emph{2}-clubs are
said to be \emph{edge connected} when they have at least one common edge.
 Two nonseparable \emph{2}-clubs $H_{0}$ and $H_{k}$ are
\emph{edge chained} if they are edge connected or there is a sequence of edge connected  $2$-clubs
$H_{0},H_{1},..,,H_{i},..H_{k}$, such that  each  consecutive pair 
$H_{i}$, $H_{i+1}$ is edge connected.\\
Thus we can state the following proposition without proof.
\begin{proposition}
\label{thm:bor}
A set of pairwise edge chained nonseparable $2$-clubs is a borough.
\end{proposition}

 An example of a borough, formed by the three edge chained \textit{2}-clubs of Figure \ref{fig:2Ctypes} is given in Figure \ref{fig:FigBorgh}.  
\begin{figure}[h]
\centering
\begin{tikzpicture}[scale=0.6]
\tikzstyle{every state}=[draw,circle,fill=white,minimum size=4pt,
                            inner sep=0pt]
\node[state] (a) at (0,4) {};
\node[state] (b) at (0,2) {};
\node[state] (c) at (0,0) {};
\node[state] (d) at (1,3) {};
\node[state] (e) at (1,1) {};
\node[state] (f) at (3,3) {};
\node[state] (g) at (3,1) {};
\node[state] (h) at (4,4) {};
\node[state] (i) at (4,2) {};
\node[state] (j) at (4,0) {};
\node[state] (k) at (6,4) {};
\node[state] (l) at (6,2) {};
\node[state] (m) at (7,5) {};
\node[state] (n) at (7,3) {};
\node[state] (o) at (7,1) {};

\path (d) edge node {} (a)
          edge node {} (b)
          edge node {} (c)
          edge node {} (e)
          edge node {} (f)

      (e) edge node {} (a)
          edge node {} (b)
          edge node {} (c)
          edge node {} (g)

      (f) edge node {} (h)
          edge node {} (i)
          edge node {} (j)

      (g) edge node {} (h)
          edge node {} (i)
          edge node {} (j)

      (h) edge node {} (i)
          edge node {} (k)

      (k) edge node {} (i)
          edge node {} (m)
          edge node {} (n)
          edge node {} (o)

      (l) edge node {} (i)
          edge node {} (m)
          edge node {} (n)
          edge node {} (o)

      (m) edge node {} (n);

\end{tikzpicture}
\caption{A Borough (Three edge-connected 2-clubs: Fig. \ref{fig:2Ctypes} :  c, a, b)}
\label{fig:FigBorgh}
\end{figure}
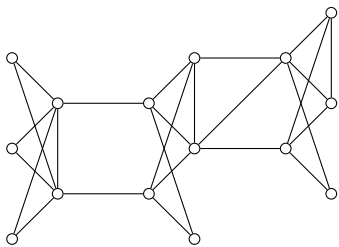
  Thus boroughs and nonseparable \textit{2}-clubs of $ G $ are formed by cycles in its basic set $  C_{\left[3,  5\right]  } = $ $ \left\lbrace C_{3}, C_{4,}, C_{5}\right\rbrace $.

\paragraph{Properties of boroughs of $G$}
 Taking into account the maximality of boroughs of \textit{G}, a number of properties follow from these definitions and previous results. These are listed below as corollaries. They show that, from a perspective of close communication, boroughs can  be seen as larger supercommunities packing or hosting close communities in a social network.\    
\begin{corollary} Given the maximality and nonseparability of boroughs of \textit{G}:
\item \textit{(i)}  Every borough of  a network $G$ is contained in exactly one  bicomponent of $G$;
\item \textit{(ii)} Two boroughs of G can not share a basic cycle of G, and each basic cycle of G is part of only one borough of G.
\item  \textit{(iii)} Each edge of a borough of \textit{G} is on a basic cycle of \textit{G} and all its basic cycles are part of that borough of \textit{G} only. Hence a basic edge of \textit{G} is a basic edge of one and only one borough of \textit{G}.
\item \textit{(iv)} Any edge of $ G $, which is not part of any borough of $ G $, is not a basic edge of \textit{G} but either a bridge or on a shortest cycle $ C_{l}$; $  l \geq\ 6$ and part of the outback of G. 
\item \textit{(v)} Thus the boroughs of \textit{G} are edge induced subgraphs of \textit{G}: they are induced by the basic edges of \textit{G}. The outback of G is induced by the non-basic edges of G.

\end{corollary}\
The maximality and non-separability of boroughs also imply:         
  \begin{corollary}  \label{ns-2C_G-B}  
 Every nonseparable \textit{2}-club of a network  $G$ is a \textit{2}-club of exactly one borough of $G$.
  \end{corollary}
 Note that a nonseparable \textit{2}-club of a graph {G} is either itself a borough of {G}, or part of just one borough of  {G}. If it would have been part of two boroughs of {G} these would share a common edge and could be merged into a larger borough, which contradicts their required maximality.\\   
However, the reverse of Corollary \ref{ns-2C_G-B} is true only for social circles and hamlets and not for nonseparable coteries, conform the following theorem: 

\begin{theorem}\label{th:HSC_B_G}

Let \textit{B} denote a borough of \textit{G}.\
\item \textit{(i)} A subgraph of G is a hamlet or social circle of G if and only if it is a hamlet or social circle of the corresponding borough B of G.\
\item \textit{(ii)} A coterie of a borough \textit{B} of \textit{G} is either itself a coterie of \textit{G} or included in a larger coterie of \textit{G}.
\end{theorem}

\begin{proof}
\textit{(i): If}: let $ B $ be a borough of $ G $ and assume that a hamlet (or social circle) of $ B $ is not a \textit{2}-club of $ G $. Then, as it is a \textit{2}-club in $ G $, it must be contained in a larger \textit{2}-club of $ G $.\\
\textit{(a)}: that \textit{2}-club cannot be a nonseparable \textit{2}-club of $ G $, because then $ B $ would not be maximal in $ G $;\\
\textit{(b)}: if that \textit{2}-club is separable, then it must be a coterie of $ G  $ with a unique central node, adjacent to all the other nodes of the \textit{2}-club (see section 3.1). But then that node must also be a node of $  B $ and the hamlet or social circle would be a coterie of $ B $ instead, contrary to assumption.\

\textit{(i): Only if}: follows directly from Corollary \ref{ns-2C_G-B}.\\

\textit{(ii)} If a coterie of a borough \textit{B} of \textit{G} is not also a coterie of \textit{G}, then it must be included in a larger 2-club of \textit{G}. That must be a coterie of \textit{G}, as by \textit{(i)} it cannot be a hamlet or social circle of \textit{G}, because then it would be included in the same social circle or hamlet of \textit{B} as well, contradicting its maximality in \textit{B}.

Thus a non-separable coterie of \textit{B} can be included in a larger, separable coterie of \textit{G} and therefore, though maximal in \textit{B} and sharing the central node, is not a coterie of {G}.\end{proof}\\

A fictitious and elementary illustration is given with Figure \ref{fig:Graph_2-Boroughs}(a) which gives an example of a simple graph \textit{G} with 29 nodes.\\ 
Considering only \textit{2}-clubs of at least three points and three lines, a straightforward count of the \textit{2}-clubs of the simple graph \textit{G} of Figure \ref{fig:Graph_2-Boroughs}(a) results in one hamlet, 5 social circles and 13 coteries of \textit{G}.\ 
The hamlet is the pentagon (C5) formed by nodes \{\textit{13,14,16,17,18}\}. The social circles are identified by:

\begin{enumerate}
\item central pairs: (\textit{w,1}), (\textit{w,6}); size: 6;
\item central pairs: (\textit{6,7}), (\textit{6,v}), (\textit{7,u}); size: 5;
\item central pairs: (\textit{20,19}), (\textit{20,17}), (\textit{19,18}); size: 5;
\item central pairs: (\textit{12,15}), (\textit{15,14}), (\textit{12,13}); size 5;
\item central pairs: (\textit{12,11}), (\textit{11,10}), (\textit{12, v}); size: 5
\end{enumerate}

Graph \textit{G} has 13 coteries, of which one coterie is the nonseparable ego-network of node \textit{3}, which as a 2-club is maximal.\

The other 12 coteries of \textit{G} are the separable ego-networks of the nodes: \textit{1, 7, 6, u, v, 17, 18, 13, 14, 12, 21} and \textit{24}.\

Note that the ego-network of node \textit{w} is nonseparable but not a 2-club of \textit{G}, because it is included in the social circle \textit{1} of \textit{G}.\\

\label{key}\begin{figure}[htbp]
\centering
\begin{subfigure}[t]{0.45\textwidth}
\includegraphics[width=\textwidth]{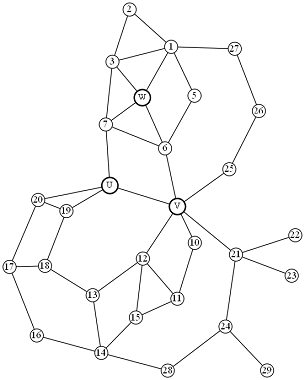}
\caption{A graph G.}
\end{subfigure}
~
\begin{subfigure}[t]{0.45\textwidth}
\includegraphics[width=\textwidth]{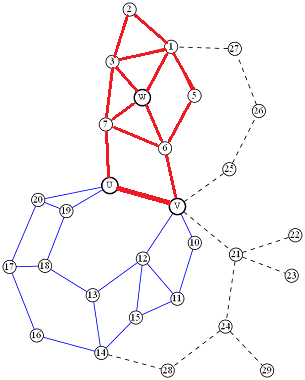}
\caption{Graph G: two boroughs and outback}
\end{subfigure}
\caption{\label{fig:Graph_2-Boroughs} Simple graph \textit{G} (29 nodes): two boroughs with two touch points and outback.\\
(Edges boroughs red-bold and blue-solid; outback black-dashed).}
\end{figure}

As shown in Figure \ref{fig:Graph_2-Boroughs}(b) graph \textit{G} has two boroughs: one indicated with red-bold edges and the other with blue-solid edges. Its \textit{outback} is given with black-dashed edges.\

If we consider only the close community area covered by these two boroughs of \textit{G}, we note that all non-separable 2-clubs of \textit{G} are contained by the two boroughs:
\begin{itemize}
\item the top red-bold borough has two social circles, which are the social circles 1 and 2 of \textit{G}, and its nonseparable coterie of node \textit{3};
\item the bottom blue-solid borough has three social circles: the social circles 3, 4 and 5 of \textit{G}, as well as its hamlet.
\end{itemize}

Boroughs of a graph G can have one or more common points, to be called \textit{touch points}, as illustrated in Figure \ref{fig:Graph_2-Boroughs}(b) by the neighbour nodes (\textit{u,v}). The extra bold red edge \textit{(u,v)} belongs to the top red-bold borough only, as it is part of its basic cycle formed by edges \textit{u-v-6-7-u}. Its other cycle \textit{u-v-12-13-18-19,u} is a hole of \textit{G} of 6 nodes and edges, and not a basic cycle of the lower blue-solid borough. Would that hole have been a basic cycle instead, then the two boroughs would have been edge chained and, due to the required maximality for boroughs of \textit{G} (see definition \ref*{Def:Borgh}), form a single borough of \textit{G}. Thus edge (\textit{u,v}) is a basic edge for the top red-bold borough only.\\
More general: if touch points of a graph \textit{G} are on common basic cycles of \textit{G}, then, due to maximality of boroughs of \textit{G}, all these basic cycles belong to the same borough of \textit{G}.\\

Referring to the particular nature of ego-networks as coteries (see conclusion Subsection 3.2 ), we see in Figure \ref*{fig:Graph_2-Boroughs} that the separable coteries of \textit{G} are not fully covered by its boroughs. For instance, the ego-network of touch points between boroughs, or between boroughs and outback of a graph, are separable and therefore coteries of that graph. This is illustrated above for the ego-networks of nodes \textit{u} and \textit{v}. Though both are coteries of \textit{G}, they are dissolved as such, because their edges are distributed over the two boroughs and, for node \textit{v}, the outback of \textit{G}. Again, the ego-networks of nodes \textit{1} and \textit{14}, reduced by missing outback edges, are also (separable) coteries of their boroughs of \textit{G}, but not of graph \textit{G} itself, because they are included in the corresponding unreduced (separable) coterie of \textit{G}.\

Moreover, the separable coteries in the outback of a graph, mainly stars, such as the ego-networks of nodes \textit{21} and \textit{24} in Figure \ref*{fig:Graph_2-Boroughs} \textit{(b)}, will be ignored by a focus on the boroughs only.\\

However: all of the coteries of the boroughs are either also coteries of \textit{G}, or included in a separable coterie of \textit{G} sharing its ego node: 

- for the red-bold borough: its three separable coteries are either also coteries of \textit{G} (see nodes \textit{6} and \textit{7}) or included in a coterie of \textit{G} \textit{i.e.} that of node \textit{1};

- for the blue-solid borough: its  5 separable coteries are also coteries of \textit{G}: those of nodes \textit{17, 18, 13} and \textit{12} or included in a coterie of \textit{G} with the same ego node: \textit{e.g.} node \textit{14}, and its only nonseparable coterie of node \textit{3} coincides with its nonseparable coterie of \textit{G}.\\

Consequently a graph or network \textit{G} can be partitioned into its boroughs and its outback, where its basic edges and their basic cycles induce the borough structure of \textit{G}, which contains its areas of \textit{close communication} and its \textit{2}-clubs as close communities, while its non-basic edges determine its \textit{outback} of more remote communication.
So, when the main focus of the analysis of a graph or network is on its close community structure, one might as well ignore its outback part and focus on its boroughs and the \textit{2}-clubs contained by them.\\

Lastly, it is well known that removal of an edge from a nonseparable graph
can increase its diameter. The next corollary shows that for boroughs this increase is limited to at most three.

\begin{corollary}
Let \textit{B} be a borough of \textit{G} and let $ B-(u,v) $, denote the subgraph obtained by removing an edge $ (u,v) $ from \textit{B}, then
\begin{equation*}
dm(B)\leq dm(B-(u,v)) \leq dm(B)+3. 
\end{equation*}

\end{corollary}

\begin{proof}
Consider the set of all shortest paths containing $ (u,v) $ defining distances between pairs of nodes of \textit{B}. Note that such pairs can also be joined by alternative shortest paths in   \textit{B} not containing \textit{(u,v)}.\\ 
As \textit{(u,v)} is on a basic cycle, its removal extends the distance between the nodes $ u $ and $ v $ by 1, 2, or 3 along the remaining part of the basic cycle(s) on \textit{(u,v)}. Thus, all distances between pairs of nodes of the set, and the diameter of $ B-(u,v) $, increase by at most 3.
\end{proof}\\
An increase of the diameter by exactly 3 implies that the removed edge is on a basic $ C_{5} $ of a hamlet of \textit{B}, as illustrated by Figure \ref{fig:min3} for the removal of the bold-lined edge.\\

\begin{figure}
\centering
\begin{tikzpicture}[scale=0.6]
\tikzstyle{every state}=[draw,circle,fill=white,minimum size=4pt,
                            inner sep=0pt]
\node[state] (a) at (0,0) {};
\node[state] (b) at (2,0) {};
\node[state] (c) at (4,0) {};
\node[state] (d) at (6,0) {};
\node[state] (e) at (0,2) {};
\node[state] (f) at (2,2) {};
\node[state] (g) at (4,2) {};
\node[state] (h) at (6,2) {};
\node[state] (i) at (3,3) {};

\path (a) edge node {} (b)
          edge node {} (e)

      (f) edge node {} (b)
          edge node {} (e)
          edge node {} (i)

      (g) edge node {} (c)
          edge node {} (h)
          edge node {} (i)

      (d) edge node {} (c)
          edge node {} (h)

      (c) edge node {} (h)

      (b) edge node {} (e)
          edge[line width = 2pt] node {} (c);
\end{tikzpicture}
\caption{Borough with diameter 3. After removing the bold edge, the remaining graph has diameter 6.}
\label{fig:min3}
\end{figure}
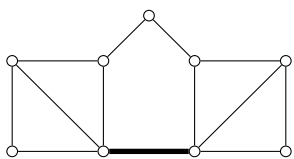
 
 \textit{Summary}. We conclude that the set $\mathcal{B}\left(G\right)$ of boroughs of $ G $ contains the proper \emph{2}-clubs of $G$, as distributed across and within its boroughs.
Thus, where we defined \emph{2-clubs} as the basic type of \emph{close
community} in a network, we can see the \emph{boroughs} to which they  belong as a
\emph{supercommunity} in that network, enveloping chained sets of such close communities.\

Moreover, conform footnote~\ref{not:k-Bor} these concepts and results can be extended to the general case of diameter $ k $. Corresponding $k-  $clubs (diameter at most $ k $) and $ k $-boroughs are then both formed from the basic set  $  C_{\left[3,  2k+1\right]  } = $ $ \left\lbrace C_{3},...,C_{2k+1}\right\rbrace  $ of basic cycles with diameter $ 3 $ to $ k $.

For instance, since the early days of social network analysis (SNA) triangles and triad censuses have been of central importance in the analysis of
social networks (\emph{e.g.} \citep{Holland1970,Holland1971,Davis1972,Johnsen1985,Frank1988,Watts1998}). Such \lq very local structures\rq, \citep{Faust2007} of direct communication between neighbours in triads, correspond to \textit{1-clubs} and \textit{1-boroughs}, as  formed by triangles only (\emph{e.g. }$ 3 $-cliques $\left(
K_{3}\right)  $ as in \citep{Palla2005}). It is not difficult to
see that the $1$-clubs and $1$-boroughs of a network are nested in the $2$-clubs and ($2$)-boroughs which are the subject of this paper.
\\
In other, \textit{e.g.} topological, contexts $ 3 $-clubs (at most diameter 3) and corresponding $ 3 $-boroughs will require for their formation the smallest 3-clubs hexagon ($ C_{6} $) and heptagon ($ C_{7} $) in the corresponding basic set $\left\lbrace  C_{3},..,C_{7}\right\rbrace  $. A rather special case is that of a 'football' type of graph, a single borough, consisting of pentagons and hexagons only, so that all its $ 3 $-clubs are formed by just two types from the basic set: the pentagon $\left(  C_{5}\right)  $ and hexagon $ \left( C_{6}\right)  $.

\section{Some applications}

With the introduction of $k$-clubs (\cite{Mokken1979}) it was pointed out that in practice their search, detection, and identification in other than small networks would be a hard, if not prohibitive, computing task, at the time beyond available hardware and algorithmic capabilities, such as the clique algorithm of Bron and Kerbosch (\cite{Bron1973}).\

Later results in computational complexity theory demonstrated that for $ k\geq2 $ several variants of $ k $-club detection were NP-hard (\textit{e.g}. \cite{Bourjolly2002}; \cite{Balasundaram2005}; ), such as, for instance, finding a $ k $-club of size larger than $\Delta(G)+1 $ (\cite{Butenko2007}), or more generally, for a given \textit{k}-club, finding a larger \textit{k}-club containing it (\cite{Pajouh2012}).

\subsection{Finding boroughs and \textit{2}-clubs}

Despite these theoretical limits, in the last decade resources and algorithm theory have made significant progress toward workaround, heuristic and practical detection algorithms to detect\space $k$-clubs \space(\citep{Bourjolly2000,Bourjolly2002,Pasupuleti2008,Asahiro2010,Yang2010,Carvalho2011,Schafer2012,Chang2012,Veremyev2012,Hartung2012,Hartung2013,Pajouh2012,Shahinpour2012}). Most of these algorithms find either $ k $-clubs of at least a given minimum size in a given graph G, or the largest (in number of nodes) $ k $-club of G.\medskip

These sources inspired us to develop some specific software modules enabling us to detect both boroughs and $ 2 $-clubs in a simple graph.
From a perspective of community detection and network analysis it is more interesting to find (inclusion-wise maximal) $ 2 $-clubs than just the largest $ 2 $-clubs
(\textit{i.e.} maximum cardinality). Hence our approach of detecting \textit{2}-clubs in a graph was designed to achieve or approximate that purpose within the limits of available computational capabilities.\

In doing so we made use of the crucial intermediate position of the boroughs, as separate components of a graph or network, hosting its edge-chained $ 2 $-clubs: its hamlets, social circles and (part of) its coteries (see Theorem \ref*{th:HSC_B_G}). This suggested a two step approach to finding all \textit{2}-clubs: first find all boroughs, then find all \textit{2}-clubs inside boroughs,  taking into account the proviso at the end of Theorem \ref*{th:HSC_B_G} concerning the special nature of coteries of a network and of its boroughs.\

We thus developed algorithms, conform to Definition \ref*{Def:Borgh}, to detect all boroughs of a graph by joining and chaining cycles from its set of basic cycles (triangles, squares and pentagons), using available methods of finding all cycles (e.g. \cite{Tiernan1970,Weinblatt1972,Fosdick1973,Johnson1975} or, more specifically, of finding only cycles of given length (\cite{Alon1994,Yuster2011}).\medskip

Another set of algorithms was developed for finding all \textit{2}-clubs of a graph, sorted by type (coterie, social circle or hamlet). 

Thus we could also detect the \textit{2}-clubs in separate selected (e.g. the largest) boroughs of a network.

The (usually many and overlapping) \textit{2}-clubs that are found are stored in a database, per borough classified according to the three possible types/level of close communication (coteries, social circles, hamlets)  and per type sorted according to size. They can then be inspected and analyzed by a Viewer interface. Details are given in \blind{\cite{Laan2012}} and in the available open source licensed package by \blind{\cite{Laan2014}}.

\subsection{Some real network results}

In this section we illustrate the concepts introduced above with some datasets, chosen to cover different data domains as well as to provide some analytic perspectives. The different data domains and associated network data are:
\begin{itemize}
\item the well-known small network of \textit{Zachary's karate club};
\item \emph{corporate board networks} as given by the interlocking directorate networks for the top 300 European firms for the year 2010;
\item \emph{co-authorship data} taken from  the large DBLP dataset. 
\end{itemize}

The Zachary data will illustrate the perspectives of \textit{2}-club analysis at the micro scale of small face-to-face networks.\

The European corporate network concerns a much larger network and the ensuing multitude of \textit{2}-clubs thus changes the analytic perspectives.

Finally we investigated the distribution of boroughs and their sizes in much larger datasets, as illustrated for the DBLP co-authorship data set. 

\subsubsection{Zachary's karate club}

Zachary's (\cite{Zachary1977}) well known data set concerned a voluntary
association, a university student karate club, with a total membership
of about 60 persons/nodes. Zachary analyzed a valued network, where edges
denoted the observed number of 8 types of mutual interaction outside
karate lessons. He restricted his analysis to the main component of
34 interacting members, thus disregarding 26 other non-connected or
isolated members. After conflicting views two factions polarized around
two main opponents - node 1 (labeled 'Mr. Hi', the karate teacher)
and node 34 ('John A.', president and main officer) - and subsequently
split accordingly. Zachary predicted the composition of the splits
using a max-flow-min-cut algorithm. 

For this paper we reanalyzed his data in the form of a simple undirected graph with an edge indicating at least one of the 8 types of interaction.
First we used a standard SNA package (\cite{Borgatti2002}) and then applied our borough and \textit{2}-club detection algorithms to the
relevant components. This simple connected network of 34 nodes has
diameter 5 and consists of two bicomponents (size 27 and 7), separated
by a common cut point (node 1: Mr. Hi), and a pendant (node 12) attached
to node 1 (Mr. Hi) as well. Both bicomponents prove to be boroughs
and node 1 (Mr. Hi) is a member, and touch point, of each of these. 
Thus the ego-network of cutpoint Mr. Hi is a coterie in the larger graph and distributed over the two boroughs and the pendant node 12.
Given the face-to-face nature of this (subset of) a student association, it is not surprising that all its edges, except the pendant (1,12), are part of at least one \textit{2}-club in just one of the two boroughs.\\
The smallest borough of size 6 contains, in addition to touch node 1 (Mr. Hi), nodes 5, 6, 7, 11, and 17, which were not further considered by Zachary. It has diameter 2 and thus is a \textit{2}-club (a social circle) and a trivial borough as such.\\
The second, largest, borough of size 27 corresponds to the network analyzed in Zachary's paper. This borough has diameter 4 and contains 13 \textit{2}-clubs including 4 coteries (all separable),
8 social circles, and one hamlet. They are listed in Table~\ref{tab:kc}. 

The two opposing leaders, Mr. Hi (node 1) and John A. (node 34), are
both part of just three \textit{2}-clubs:
\begin{itemize}
\item the 7-node ego-network (coterie) of node 32;
\item a social circle of size 14 (with central pair 34-14);  and 
\item the hamlet of size 8.
\end{itemize}
The latter two \textit{2}-clubs are depicted in Figure~\ref{fig:HiClubs}.\\

Moreover, each of the two opposing nodes (1 or 34) are part of five \textit{2}-clubs
excluding the other opponent. Hence, all \textit{ 2}-clubs contain at least
one of the two opponents: node 1 or node 34. 
In particular the 14 node social circle and 8 node hamlet look like
negotiation forums of the two opposing sides. For instance, the hamlet
of 8 nodes connects the central egos (nodes 1, 34, 3, and 32) of
the 4 coteries. 

\begin{table}
\begin{footnotesize}
\begin{tabular}{lrl}
\hline
\textbf{Type}& \textbf{Size} & \textbf{List of Members} \\ \hline
Coterie&7&[\textbf{1}, 25, 26, 29, 32, 33, \textbf{34}]
\\Social circle&8&[24, 26, 28, 29, 30, 32, 33, \textbf{34}]
\\Social circle&8&[3, 24, 25, 28, 29, 32, 33, \textbf{34}]
\\Social circle&8&[24, 25, 26, 28, 29, 32, 33, \textbf{34}]
\\Hamlet&8&[\textbf{1}, 3, 25, 28, 29, 32, 33, \textbf{34}]
\\Social circle&10&[\textbf{1}, 2, 3, 4, 8, 9, 14, 29, 32, 33]
\\Social circle&10&[\textbf{1}, 2, 3, 4, 8, 9, 14, 31, 32, 33]
\\Coterie&11&[\textbf{1}, 2, 3, 4, 8, 9, 10, 14, 28, 29, 33]
\\Social circle&11&[\textbf{1}, 2, 3, 4, 8, 9, 14, 18, 20, 22, 31]
\\Coterie&12&[\textbf{1}, 2, 3, 4, 8, 9, 13, 14, 18, 20, 22, 32]
\\Social circle&14&[\textbf{1}, 2, 3, 4, 9, 10, 14, 20, 28, 29, 31, 32, 33, \textbf{34}]
\\Social circle&17&[3, 9, 10, 14, 15, 16, 19, 21, 23, 24, 28, 29, 30, 31, 32, 33, \textbf{34}]
\\Coterie&18&[9, 10, 14, 15, 16, 19, 20, 21, 23, 24, 27, 28, 29, 30, 31, 32, 33, \textbf{34}]
\\\hline
\end{tabular}
\end{footnotesize}
\caption{\label{tab:kc} List of all  2-Clubs in the large 27-node Borough   of  Zachary's Karate Club.
For each two club, the type, size and the members are given. The leaders of the two parts after the split (1 and 34) are indicated in boldface.} 
\end{table}

\begin{figure}[htbp]
\centering
\begin{subfigure}[t]{0.45\textwidth}
\includegraphics[width=\textwidth]{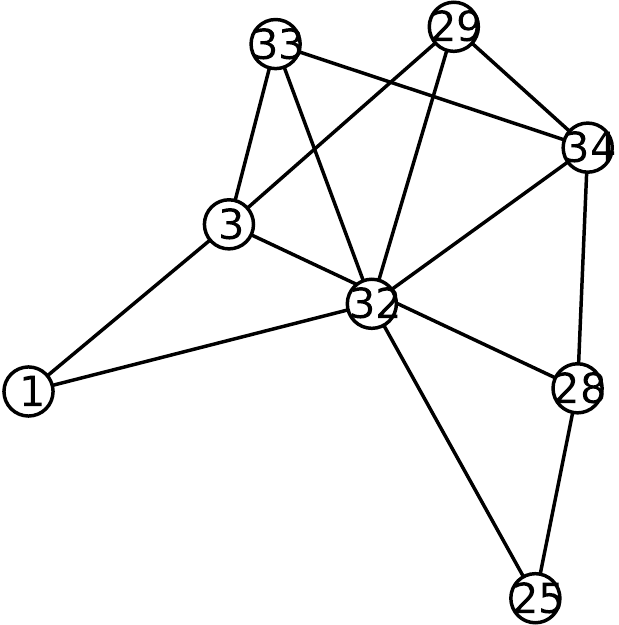}
\caption{The Hamlet.}
\end{subfigure}
~
\begin{subfigure}[t]{0.45\textwidth}
\includegraphics[width=\textwidth]{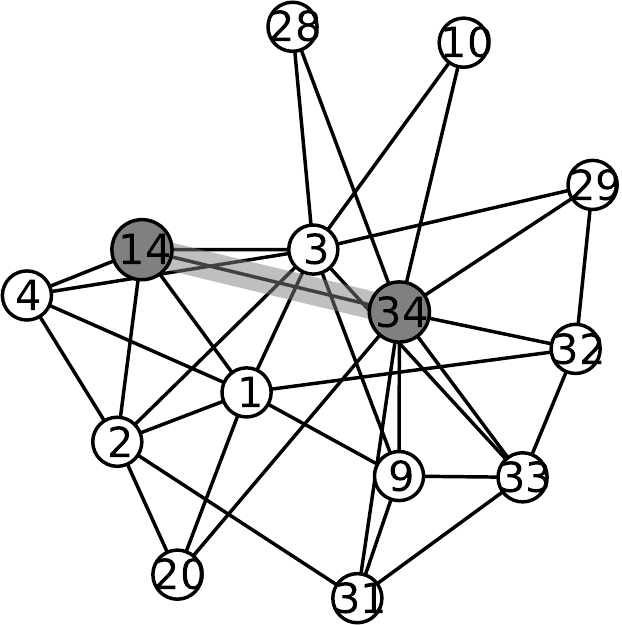}
\caption{The Social circle. The central pair is highlighted.}
\end{subfigure}
\caption{Two of the three 2-Clubs containing both John and Mr. Hi.}
%
\label{fig:HiClubs}
\end{figure}
 
These four coteries represent the ego-networks of the two opposing nodes 1 (Mr Hi: size 12) and 34 (John 
a.: size 18), and the nodes 32 (size 7) and 3 (size 11), where node 3 appears to be a supporting 'lieutenant' node for Mr. Hi (node 1) and node 32 for his opponent John A. In terms of their \textit{2}-club memberships both nodes show extensive liaison connections with the opposing side.\\

Membership of particular \textit{2}-clubs appears to be a good predictor for
faction membership after the split. To keep within the bounds of this
paper, we can illustrate this with the problematic mysterious node
9, the only node mentioned explicitly by Zachary in his paper, apart
from Mr. Hi (1) and John A. (34). Node 9 was problematic in the sense
that he was classed as a (mild) supporter of the side of 34 (John A.), but in the end showed up as a member of the opposing faction of Mr. Hi after the split. However, this move can be understood by an analysis of his \textit{2}-club memberships.\\
Node 9 was member of eight \textit{2}-clubs, each of which included at least
one of the two opponents 1 or 34, and distributed as follows:

\begin{itemize}
\item five \textit{2}-clubs with only node 1 (Mr. Hi);
\item two \textit{2}-clubs with only node 34 (John A.); and 
\item one \textit{2}-club with node 1 and node 34.
\end{itemize}
Moreover, in 7 of its eight \textit{2}-clubs node 9 is accompanied by node
3, its neighbor and firm Mr. Hi supporter, as we noted above.

So, on the basis of its \textit{2}-club affiliations alone, one would have
predicted node 9 to move (or stay) with the faction of Mr. Hi after
the split, as in fact he did. The \textit{2}-club analysis also revealed liaison roles of
two nodes, not mentioned as such in the Zachary paper: node 3 for
Mr. Hi (node 1) and node 32 for John A. (node 34).

\subsubsection{European corporate network 2010}

This network was constructed from the interlocking directorates of the largest 286 stock listed  companies, as studied by Heemskerk \citep{Heemskerk2013}.\footnote{The  European data for 2010 were kindly made available to us by Eelke Heemskerk.} 
Its nodes designate the boards of individual companies and its edges indicate that the companies they connect share at least one common director in their boards. Hence the network provides channels of interpersonal contact and communication between companies at the level of their boards.\\

The source network for these 286 companies had one giant component of 259, which we chose for further analysis.
Apart from three trivial 'boroughs' of sizes 4, 3 and 3, we found one single giant borough of 225 companies.\\

This single borough, covering 87\% of the dominant component and 79\% of all firms, formed a compact sub-network in the corporate European network of 2010, consisting of 2128 overlapping or edge-chained \textit{2}-clubs of corporations, with a median \textit{2}-club size of 10 corporations in a size range of 4-27. This result confirms Heemskerk's (2013) original conclusion that by 2010 this European network appeared to be well integrated. However, with a  diameter of 7 the borough was rather stretched.\\

\begin{table}[htbp]
\begin{tabular}{lllll}
\hline\\
\textit{European borough 2010} & \textit{Coterie} & \textit{Soc. Circ.} & \textit{Hamlet} & \textit{Total}  \\\hline
 $ $ $ $ \\ 
All \textit{}2-Clubs Borough & 138 & 717 & 1273 & 2128 \\ 
\% of  total \textit{2}-Clubs Borough & 6.5\% & 33.7\% & 59.8\% & 100\% \\ 
Size range  & 4-27 & 5-25 & 5-24 & 4-27 \\ 
Median size & 10 & 14 & 16 & 15 \\
Coverage nodes borough $ ^{*)} $ & 99.6\% & 89.4\% & 92.5\% & 100\% \\
 &  &  &  &  \\
\textit{Compagnie Nationale \`{a} Portfeuille SA} & 15 & 284 & 691 & 990 \\
\% of type \textit{2}-club borough &10.9\% & 39.6\% & 54.3\% & 46.5\% \\  
\textit{BNP Paribas SA} & 12 & 105 & 350 & 467 \\
\% of type \textit{2}-club borough &8.7\% & 14.6\% & 27.5\% & 21.9\% \\
 &  &  &  &  \\
\textit{Cie Nat. \`{a} Portefeuille} or \textit{BNP Paribas} & 25 & 332 & 752 & 1109 \\ 
\% of  total \textit{2}-Clubs Borough & 18.1\% & 46.3\% & 59.1\% & 52.1\% \\ 
Common 2-Clubs & 2 & 57 & 289 & 348 \\ 
\% \textit{Cie Nat. \`{a} Portefeuille} of \textit{BNP Paribas} & 16.7\% & 54.3\% & 82.6\% & 74.5\% \\
 &  &  &  &\\
$ ^{*)} $ {\footnotesize Coverage: \% of nodes of borough (100\% = 225)}  &  &  &  &\\
\hline 
\end{tabular}
\caption{European corporate borough 2010: 2-clubs of borough and two of its major companies compared.}
\label{Bor2Comps}
\end{table}

The multitude of \textit{2}-clubs, to be expected for large networks, can be analyzed by means of views and selections from their database. \textit{Table} \ref{Bor2Comps} shows some results.\\

The first \textit{upper} part of \textit{Table} \ref{Bor2Comps} gives the distribution of these 2128 \textit{2}-clubs of the European borough over the three types and levels of close communication.\\

The \textit{first}, most local level of communication, the \textit{coteries}, formed 6.5\% of the \textit{2}-clubs of the borough. They were the ego-networks of 138 central companies: 62\% of the 225 companies in the borough. Together the nodes of these coteries covered practically all (coverage: 99.9\%) of the companies (nodes) of the borough.\

The ego-networks of the other 85 companies were not coteries of the borough but were included in or split over larger 2-clubs.  That was, for example, the case for the German automotive company \textit{Volkswagen AG} and two large banks: the German \textit{Deutsche Bank AG} and the Spanish \textit{Banco Santander SA}.\

Thus this set of 138 coteries formed the local backbone of the borough, as the two next levels of more extended close communication, social circles and hamlets, are formed from parts of the ego-networks of their central companies. 
Their composition strongly suggests a predominance of French \textit{2}-clubs in the borough: the 20 largest coteries consist of the ego-networks of 12 French, 3 German, 2 British firms, and a Swedish, a Belgian and a Swiss company.\\

The \textit{second} intermediate level of\textit{ social circles} was formed by one third (717: 33.7\%) of the \textit{2}-clubs of the borough, with a median size of 14 companies in a size range of 5-25. Together the social circles cover 89.4\% of all nodes of the borough. \\
The composition of the largest social circles again confirms the predominance of the largest French companies in the network. They were formed around one or more central pairs of major French companies (\textit{i.e.} pairs of central neighbors adjacent to all others in the \textit{2}-club).\

In \textit{Figure} \ref{fig:SCTotSuezCieNat} a detail is given of one large social circle of 25 companies, mainly French, with 19 French and one Franco-Belgian, 2 British, 2 Dutch firms and a Luxembourg company. It shows its densest part around its two central pairs, one formed by the French company \textit{Total SA} with the French firm \textit{GDF Suez SA} and the other by \textit{Total SA} and the Belgian company \textit{Compagnie Nationale \`{a} Portefeuille SA}.\\

The \textit{third} and widest level of close communication, the \textit{hamlets}, occupied a major part (1273: 59.8\%) of the \textit{2}-clubs of the European borough, with a median size of 16 in a size range of 5-24. Altogether the hamlets of the borough cover nearly all \textit{i.e.} 92.5\% of its 225 nodes. An example is given by \textit{Figure} \ref{fig:HamlABB} to which we will return later.\\

Heemskerk (2013, p. 91) cites \textit{Compagnie Nationale \`{a} Portefeuille SA}, a Belgian investment holding of the Fr\`{e}re family, as most involved in European interlocks, with 17 European and 2 national (Belgian) interlocks.
We investigate this conclusion further in terms of its participation in major \textit{2}-clubs of the European borough of corporate interlocks.\\
\begin{figure}
\centering
\includegraphics[width=0.6\linewidth]{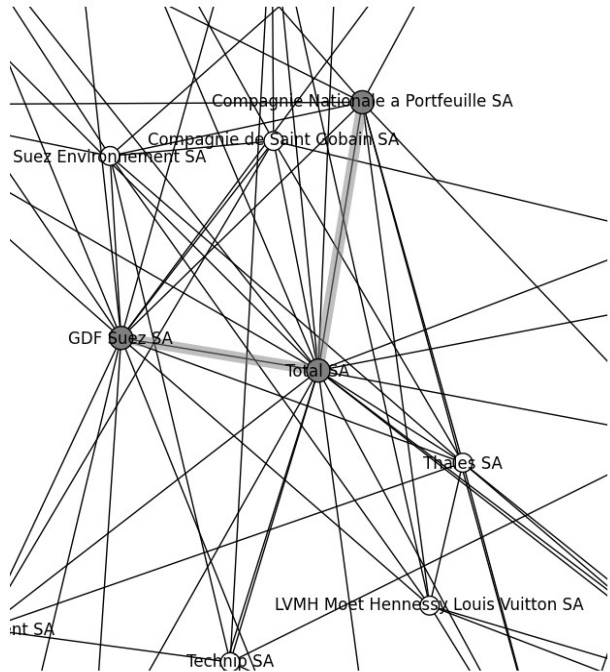}
\caption{European corporate Borough 2010: detail of center of largest Social Circle (size 25) with the two central pairs shaded.}
\label{fig:SCTotSuezCieNat}
\end{figure}\\
\
The second,  \textit{lower} part of \textit{Table 2} summarizes this analysis. It was a member of almost half (990: 46.5 \%) of the \textit{2}-clubs of the  borough.\
At the most local level it was a member of 15 (10.9\%) of the coteries of the borough. Of these 15 ego-networks,  identified by their central (ego) company and ordered by size, it was the center of the sixth coterie (size 22). Of the other 14 coteries ten were large French companies,two were Luxembourg based, followed by another Belgian company  and a German company.

This predominant francophone orientation suggests that it was more part of the French  regional network than of a cross-European one. Widening the level of close communication to its (\textit{Compagnie Nationale \`{a} Portefeuille}) membership of social circles and hamlets of the borough confirmed this impression: it was part of 284 social circles (39.6\%)and 691 (54.3\%

) hamlets. Among the largest social circles it formed part of one or more central pairs with the the five largest French companies, as illustrated by Figure \ref{fig:SCTotSuezCieNat}, where it forms one of the two central pairs with the French company \textit{Total SA}.\\

We therefore studied its \textit{2}-club memberships together with those of the largest French bank: \textit{BNP Parisbas SA}, as summarized  in the lower part of \textit{Table} \ref{Bor2Comps}. \textit{BNP Paribas SA} itself was included in 467 (21.9\%) of the \textit{2}-clubs of the borough.  
The combined membership of \textit{Compagnie Nationale \`{a} Portefeuille} or \textit{BNP Parisbas} accounted for 1109 (52.1\%), or more than half of the 2128
 \textit{2}-clubs in the European borough. In 348 of those they participated together. 
Consequently \textit{Compagnie Nationale \`{a} Portefeuille} participated in almost three quarter (74.5\%) of the \textit{2}-clubs to which  \textit{BNP Parisbas} belonged.\

Hence in terms of \textit{2}-club memberships the Belgian \textit{Compagnie Nationale \`{a} Portefeuille} was clearly a part of the center of the French corporate sub-network in 2010.
\begin{figure}
\centering
\includegraphics[width=0.5\linewidth]{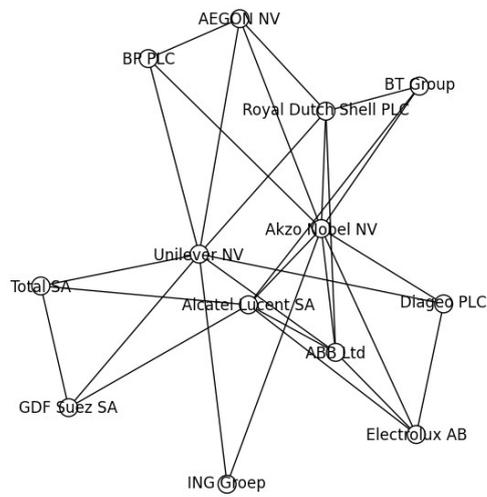}
\caption[short]{Hamlet (size 13): Swiss ABB Ltd with 12 French, British, Dutch or Swedish firms.}
\label{fig:HamlABB}
\end{figure}
\\
Subsequent developments appear to support this conclusion.  The controlling Belgian holding \textit{ERBE}, for 53\% owned by the Belgian Fr\`{e}re family and for 47\% by \textit{BNP Parisbas}, removed  \textit{Compagnie Nationale \`{a} Portefeuille}  from the Belgian stock exchange on 2 May 2011, after a succesful bid for outstanding stock. This appeared to be part of a familial succession strategy and an agreement allowing \textit{BNP Paribas} to withdraw from \textit{ERBE}. In a press release of December 10, 2013 \textit{BNP Paribas} announced its completion of this arrangement through the purchase by the \textit{Fr\`{e}re Group} of the entire \textit{BNP Paribas} shareholding in \textit{ERBE}.\\

As the second firm, most involved in European interlocks, with a reported 14 European and 2 national interlocks,  Heemskerk (\textit{l.c}) cites \textit{ABB Ltd}, a Swiss based multinational corporation operating mainly in power and automation technology, such as robotics.
In this case his conclusion appears to be fully supported by investigating its \textit{2}-club memberships in the corporate European borough for 2010.
Not surprisingly it was central ego of a coterie (size 17) of the borough, consisting  of companies from six European nations: 3 German, 2 French, 3 Swiss, 5 Swedish, 3 Dutch and 1 Finnish. 

\textit{ABB Ltd} participated in 64 social circles (size 7-21): the first 5 largest social circles solidly German with central pairs from the largest German companies. After those follow a number of mainly Swedish social circles and a number of social circles of mixed nationality.

At the widest level of close communication \textit{ABB Ltd} participated in 75 hamlets (size 5-19) of different nationalities. An example is given with the hamlet of  Figure \ref{fig:HamlABB}, containing thirteen firms: three French, four British, one Swedish and one Swiss, ABB Ltd itself.\\

For a more elaborate analysis, beyond the scope of this paper, we refer to \citep{Mokken2015b}

\subsubsection{Boroughs in DBLP co-authorship networks}
We use the DBLP database of Computer Science publications to obtain some insights on the availability of boroughs in large real live networks\footnote{Downloaded from \url{http://dblp.uni-trier.de/xml} at 2012-02-21.}. DBLP can be seen as a bipartite network consisting of authors and publications as nodes connected by the   \lq is author of\rq \space relation.
For a given integer threshold  $t$, we  induce  an undirected  co-authorship network  between authors from the DBLP network by relating two author nodes when they have coauthored at least $t$ publications.
For $t$ between 5 and 10, Table~\ref{DBLP} contains basic statistics about the number of boroughs, and their distribution according to their size. Thus, for a threshold $t$, the number of nodes in the $t$-co-authorship network is the number of authors who have at least $t$ joint publications with one other author. The density of the resulting networks is fairly stable and slowly increases from $2\cdot 10^{-5}$ for $t=5$ to $4\cdot 10^{-5}$ for $t=10$.\\

We can draw several conclusions from this experiment:
\begin{table}
\begin{footnotesize}
\begin{center}
\begin{tabular}{r rr rrr}
\hline
\textbf{$t$}& \textbf{\#nodes} & \textbf{\#edges} &   \textbf{\#Boroughs} & \textbf{5  largest sizes }\\ 
\hline
5&152018&239309&15550&[23414,211,91,78,74]\\ 
6&118423&169328&12177&[8054,153,81,51,50]\\ 
7&94784&125308&9713&[6294,144,51,50,39]\\ 
8&77288&95270&7701&[4741,137,90,56,46]\\ 
9&63916&74185&6115&[3743,121,45,42,40]\\ 
10&53596&59071&4913&[2907,118,43,36,34]\\ 
\\
\hline
\end{tabular}
\end{center}
\end{footnotesize}
\caption{\label{DBLP}
Basic statistics about the Borough distribution in the co-authorship networks based on DBLP, for thresholds on the number of coauthored publications ranging from $ t= $ 5 to 10.}
\end{table}
\begin{itemize}
\item Large sparse networks contain relatively many boroughs,  roughly an order of magnitude smaller than the number of nodes in the network.
\item Large networks contain one \lq giant\rq \space borough, whose number of nodes is roughly an order of magnitude smaller than the number of nodes in the network.
\item    The number of nodes of all boroughs except the \lq giant\rq \space one is small.
\item The sizes of the boroughs of the DBLP co-authorship networks are distributed according to a power-law. 
\end{itemize}
Given the complexity of finding basic cycles for such large networks and available capacities, we only computed the boroughs for $t$ from 5 to 10.

\section{Discussion}
 
 Our reanalysis of the Zachary karate club network demonstrated the usefulness of \textit{2}-club analysis for relatively small 'very local' networks \citep{Faust2007}. The next example, concerning the networks of corporate interlocking directorates for Europe in 2010, illustrated the huge numbers of distinct, but overlapping \textit{2}-clubs of the three types to  be expected for larger, possibly dense networks. However, once the boroughs and their \textit{2}-clubs are detected, identified and stored, the challenge of their analysis can be met with the plethora of currently available statistical methods of search, data mining and matching techniques of massive databases.\\
 
 Our exercise with the large DBLP data set shows that a much larger challenge will be how to combine the micro, \textit{i.e} very local, in-depth focus of close communication by boroughs and \textit{2}-clubs with the global analysis of the Big Data massive networks which currently confront community detection. Promising techniques can be based on the analysis of appropriate segments of such networks, using their hierarchical modularities with techniques such as proposed by \cite{Blondel2008} or by focusing on selected 2-neighborhoods.\\
 
 Finally, some researchers (\textit{e.g.} \citep{Hartung2012}) have noted that the largest \textit{2}-clubs they found in real-world networks just coincided with the ego-network of a node with maximum degree ($\Delta\left(G\right)+1$). As the size of a coterie cannot be larger than that limit, any \textit{2}-club of larger size than the maximum degree plus one must be a hamlet or a social circle.
In our analyses of various real world networks we also did not find a social circle or hamlet larger than the maximum degree, the largest coterie.\

 As it is not difficult to construct examples of networks with hamlets or social circles which exceed that limit, a question of further research is to hunt for empirical, real-world datasets where that is indeed the case. Networks with no or limited preferential attachment seem likely candidates.\\
 
\blind{\section*{Acknowledgment}
We thank Johan van Doornik for several suggestions in the initial stages of our research and are grateful to some referees for suggestions for improvement.}

 
\bibliographystyle{elsarticle-harv}  
\bibliography{Close-Communities}

 \end{document}